\documentclass{llncs}
\usepackage{amsmath}
\usepackage{makeidx}
\usepackage{epic}

\input{tcilatex}
\begin{document}

\title{Full Square Rhomboids and Their Algebraic Expressions}
\author{Mark Korenblit \\
\institute{Holon Institute of Technology, Israel\\
korenblit@hit.ac.il}}
\maketitle

\begin{abstract}
The paper investigates relationship between algebraic expressions and
graphs. We consider a digraph called a full square rhomboid that is an
example of non-series-parallel graphs. Our intention is to simplify the
expressions of full square rhomboids and eventually find their shortest
representations. With that end in view, we describe two decomposition
methods for generating expressions of full square rhomboids and carry out
their comparative analysis.\medskip\ 
\end{abstract}

\section{Introduction\label{intro}}

A \textit{graph }$G=(V,E)$ consists of a \textit{vertex set\ }$V$ and an 
\textit{edge set\ }$E$, where each edge corresponds to a pair $(v,w)$ of
vertices. A graph $G^{\prime}=(V^{\prime},E^{\prime})$ is a \textit{subgraph}
of $G=(V,E)$ if $V^{\prime}\subseteq V$ and $E^{\prime}\subseteq E$. A graph 
$G$ is a \textit{homeomorph} of $G^{\prime}$ if $G$ can be obtained by
subdividing edges of $G^{\prime}$ with new vertices. We say that a graph $%
G^{2}=(V,E^{\prime})$ is a \textit{square of a graph} $G=(V,E)$ if $%
E^{\prime}=\left\{ (u,w):(u,w)\in E\vee\left( (u,v)\in E\wedge(v,w)\in
E\right) \text{ for some }v\in V\right\} $. A two-terminal directed acyclic
graph (\textit{st-dag}) has only one source and only one sink.

We consider a \textit{labeled graph} which has labels attached to its edges.
Each path between the source and the sink (a \textit{sequential path}) in an
st-dag can be presented by a product of all edge labels of the path. We
define the sum of edge label products corresponding to all possible
sequential paths of an st-dag $G$ as the \textit{canonical expression }of $G$%
. An algebraic expression is called an \textit{st-dag expression} (a \textit{%
factoring of an st-dag} in \cite{BKS}) if it is algebraically equivalent to
the canonical expression of an st-dag. An st-dag expression consists of
literals (edge labels), and the operators $+$ (disjoint union) and $\cdot$
(concatenation, also denoted by juxtaposition). An expression of an st-dag $%
G $ will be hereafter denoted by $Ex(G)$.

We define the total number of literals in an algebraic expression as the 
\textit{complexity of the algebraic expression}. An equivalent expression
with the minimum complexity is called an \textit{optimal representation of
the algebraic expression}.

A \textit{series-parallel} \textit{graph} is defined recursively so that a
single edge is a series-parallel graph and a graph obtained by a parallel or
a series composition of series-parallel graphs is series-parallel. As shown
in \cite{BKS} and \cite{KoL}, a series-parallel graph expression has a
representation in which each literal appears only once. This representation
is an optimal representation of the series-parallel graph expression. For
example, the canonical expression of the series-parallel graph presented in
Figure \ref{fig1} is $abd+abe+acd+ace+fe+fd$. Since it is a series-parallel
graph, the expression can be reduced to $(a(b+c)+f)(d+e)$, where each
literal appears once.

\begin{figure}[tbp]
\setlength{\unitlength}{0.7cm}
\par
\par
\begin{picture}(5,3.5)(-3,0.3)\thicklines
\multiput(1,3)(3,0){4}{\circle*{0.15}}
\put(1,3){\vector(1,0){3}} \put(2.5,3.3){\makebox(0,0){$a$}}
\qbezier(4,3)(5.5,5)(7,3) \put(7.085,3){\vector(3,-2){0}}
\put(5.5,4.3){\makebox(0,0){$b$}}
\qbezier(4,3)(5.5,1)(7,3) \put(7.085,3){\vector(3,2){0}}
\put(5.5,2.3){\makebox(0,0){$c$}}
\qbezier(7,3)(8.5,5)(10,3) \put(10.085,3){\vector(3,-2){0}}
\put(8.5,4.3){\makebox(0,0){$d$}}
\qbezier(7,3)(8.5,1)(10,3) \put(10.085,3){\vector(3,2){0}}
\put(8.5,2.3){\makebox(0,0){$e$}}
\qbezier(1,3)(4,-2)(7,3) \put(7.085,3){\vector(4,3){0}}
\put(4,0.8){\makebox(0,0){$f$}}
\end{picture}
\caption{A series-parallel graph.}
\label{fig1}
\end{figure}
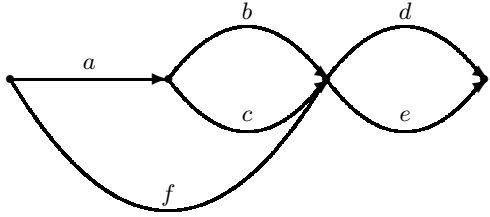

A \textit{Fibonacci graph} \cite{GoP} has vertices $\{1,2,3,\ldots ,n\}$ and
edges $\{\left( v,v+1\right) \mid v=1,2,\ldots ,n-1\}\cup \left\{ \left(
v,v+2\right) \mid v=1,2,\ldots ,n-2\right\} $. As shown in \cite{Duf}, an
st-dag is series-parallel if and only if it does not contain a subgraph
which is a homeomorph of the \textit{forbidden subgraph} positioned between
vertices $1$ and $4$ of the Fibonacci graph illustrated in Figure \ref{fig2}%
. Thus a Fibonacci graph gives a generic example of non-series-parallel
graphs. 
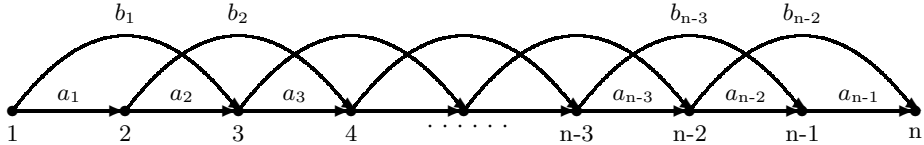
\begin{figure}[tbph]
\setlength{\unitlength}{1.0cm}
\par
\par
\begin{picture}(5,1.2)(0,-0.3)\thicklines
\multiput(0,0)(1.5,0){9}{\circle*{0.15}}
\put(0,-0.3){\makebox(0,0){1}}
\put(1.5,-0.3){\makebox(0,0){2}}
\put(3,-0.3){\makebox(0,0){3}}
\put(4.5,-0.3){\makebox(0,0){4}}
\put(7.5,-0.3){\makebox(0,0){n-3}}
\put(9,-0.3){\makebox(0,0){n-2}}
\put(10.5,-0.3){\makebox(0,0){n-1}}
\put(12,-0.3){\makebox(0,0){n}}
\multiput(0,0)(1.5,0){8}{\vector(1,0){1.5}}
\put(0.75,0.2){\makebox(0,0){$a_{1}$}}
\put(2.25,0.2){\makebox(0,0){$a_{2}$}}
\put(3.75,0.2){\makebox(0,0){$a_{3}$}}
\put(8.25,0.2){\makebox(0,0){$a_{\text{n-3}}$}}
\put(9.75,0.2){\makebox(0,0){$a_{\text{n-2}}$}}
\put(11.25,0.2){\makebox(0,0){$a_{\text{n-1}}$}}
\qbezier(0,0)(1.5,2)(3,0)
\qbezier(1.5,0)(3,2)(4.5,0)
\qbezier(3,0)(4.5,2)(6,0)
\qbezier(4.5,0)(6,2)(7.5,0)
\qbezier(6,0)(7.5,2)(9,0)
\qbezier(7.5,0)(9,2)(10.5,0)
\qbezier(9,0)(10.5,2)(12,0)
\multiput(3.085,0)(1.5,0){7}{\vector(3,-2){0}}
\put(1.5,1.3){\makebox(0,0){$b_{1}$}}
\put(3,1.3){\makebox(0,0){$b_{2}$}}
\put(9,1.3){\makebox(0,0){$b_{\text{n-3}}$}}
\put(10.5,1.3){\makebox(0,0){$b_{\text{n-2}}$}}
\multiput(5.55,-0.2)(0.2,0){6}{\circle*{0.02}}
\end{picture}
\caption{A Fibonacci graph.}
\label{fig2}
\end{figure}

Mutual relations between graphs and expressions are discussed in \cite{BKS}, 
\cite{GoM}, \cite{GMR}, \cite{KoL}, \cite{Mun1}, \cite{Mun2}, \cite{Nau}, 
\cite{SaW}, and other works. Specifically, \cite{Mun1}, \cite{Mun2}, and 
\cite{SaW} consider the correspondence between series-parallel graphs and
read-once functions. A Boolean function is defined as \textit{read-once} if
it may be computed by some formula in which no variable occurs more than
once (\textit{read-once formula}). On the other hand, a series-parallel
graph expression can be reduced to the representation in which each literal
appears only once. Hence, such a representation of a series-parallel graph
expression can be considered as a read-once formula (boolean operations are
replaced by arithmetic ones).

An expression of a homeomorph of the forbidden subgraph belonging to any
non-series-parallel st-dag has no representation in which each literal
appears once. For example, consider the subgraph positioned between vertices 
$1$ and $4$ of the Fibonacci graph shown in Figure \ref{fig2}. Possible
optimal representations of its expression are $a_{1}\left(
a_{2}a_{3}+b_{2}\right) +b_{1}a_{3}$ or $\left( a_{1}a_{2}+b_{1}\right)
a_{3}+a_{1}b_{2}$. For this reason, an expression of a non-series-parallel
st-dag can not be represented as a read-once formula. However, for arbitrary
functions, which are not read-once, generating the optimum factored form is
NP-complete \cite{Wan}. Some algorithms developed in order to obtain good
factored forms are described in \cite{GoM}, \cite{GMR} and other works. In 
\cite{KoL} we presented an algorithm,\ which generates the expression of $%
O\left( n^{2}\right) $ complexity for an $n$-vertex Fibonacci graph.

In \cite{KoL4} we considered a non-series-parallel st-dag called a \textit{%
square rhomboid} (Figure \ref{rhom_fig12}). This graph looks like a planar
approximation of the square of a \textit{rhomboid}, which is a series
composition of \textit{rhomb} graphs. A square rhomboid consists of the same
vertices as the corresponding rhomboid. However, edges labeled by letters $a$%
, $b$, and $c$ (see Figure \ref{rhom_fig12}) are absent in a rhomboid.
Geometrically, a square rhomboid ($SR$ for brevity) can be considered to be
a \textquotedblright gluing\textquotedblright\ of two Fibonacci graphs (the
upper one consists of edges labeled by $e$, $b$, $c$ and the lower one
consists of edges labeled by $d$, $b$, $a$), i.e., it is the next harder one
in a sequence of increasingly non-series-parallel graphs.

\begin{figure}[tbph]
\setlength{\unitlength}{0.8cm}
\par
\par
\begin{picture}(10,3)(1.7,-1.6)\thicklines
\multiput(0,0)(3,0){7}{\circle*{0.1875}}
\multiput(1.5,1.5)(3,0){6}{\circle*{0.1875}}
\multiput(1.5,-1.5)(3,0){6}{\circle*{0.1875}}
\put(0,0.35){\makebox(0,0){1}} \put(3.02,0.35){\makebox(0,0){2}}
\put(15,0.5){\makebox(0,0){n-1}} \put(18,0.35){\makebox(0,0){n}}
\put(1.5,1.8){\makebox(0,0){1}} \put(4.5,1.8){\makebox(0,0){2}}
\put(13.5,1.8){\makebox(0,0){n-2}}
\put(16.5,1.8){\makebox(0,0){n-1}}
\put(1.5,-1.8){\makebox(0,0){1}} \put(4.5,-1.8){\makebox(0,0){2}}
\put(13.5,-1.8){\makebox(0,0){n-2}}
\put(16.5,-1.8){\makebox(0,0){n-1}}
\multiput(0,0)(3,0){6}{\vector(1,1){1.5}}
\multiput(0,0)(3,0){6}{\vector(1,-1){1.5}}
\multiput(1.5,1.5)(3,0){6}{\vector(1,-1){1.5}}
\multiput(1.5,-1.5)(3,0){6}{\vector(1,1){1.5}}
\multiput(0,0)(3,0){6}{\vector(1,0){3}}
\multiput(1.5,1.5)(3,0){5}{\vector(1,0){3}}
\multiput(1.5,-1.5)(3,0){5}{\vector(1,0){3}}
\put(1.5,0.3){\makebox(0,0){$b_{1}$}}
\put(16.5,0.3){\makebox(0,0){$b_{\text{n-1}}$}}
\put(3,1.75){\makebox(0,0){$c_{1}$}}
\put(15,1.75){\makebox(0,0){$c_{\text{n-2}}$}}
\put(3,-1.75){\makebox(0,0){$a_{1}$}}
\put(15,-1.75){\makebox(0,0){$a_{\text{n-2}}$}}
\put(0.55,0.95){\makebox(0,0){$e_{1}$}}
\put(15.4,0.95){\makebox(0,0){$e_{\text{2n-3}}$}}
\put(2.4,0.95){\makebox(0,0){$e_{2}$}}
\put(17.56,0.95){\makebox(0,0){$e_{\text{2n-2}}$}}
\put(0.6,-0.95){\makebox(0,0){$d_{1}$}}
\put(15.4,-0.95){\makebox(0,0){$d_{\text{2n-3}}$}}
\put(2.4,-0.95){\makebox(0,0){$d_{2}$}}
\put(17.56,-0.95){\makebox(0,0){$d_{\text{2n-2}}$}}
\multiput(4,0.2)(0.2,0){6}{\circle*{0.02}}
\multiput(13,0.2)(0.2,0){6}{\circle*{0.02}}
\multiput(5.5,1.7)(0.2,0){6}{\circle*{0.02}}
\multiput(11.5,1.7)(0.2,0){6}{\circle*{0.02}}
\multiput(5.5,-1.7)(0.2,0){6}{\circle*{0.02}}
\multiput(11.5,-1.7)(0.2,0){6}{\circle*{0.02}}
\end{picture}
\caption{A square rhomboid of size $n$.}
\label{rhom_fig12}
\end{figure}
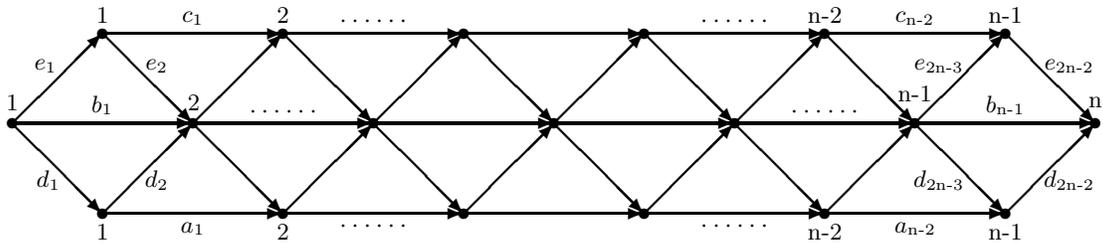

In this paper we investigate a more complicated graph called a \textit{full
square rhomboid} ($FSR$) which is a real square of a rhomboid and, in
addition to all edges of an $SR$, has edges labeled by $f$ and $g$ (Figure %
\ref{fsr_fig1}).

\begin{figure}[tbph]
\setlength{\unitlength}{0.8cm}
\par
\fontsize{8}{0}\selectfont%
\par
\begin{picture}(10,3)(1.7,-1.6)\thicklines
\multiput(0,0)(3,0){7}{\circle*{0.1875}}
\multiput(1.5,1.5)(3,0){6}{\circle*{0.1875}}
\multiput(1.5,-1.5)(3,0){6}{\circle*{0.1875}}
\put(0,0.35){\makebox(0,0){1}} \put(3.02,0.35){\makebox(0,0){2}}
\put(15,0.45){\makebox(0,0){n-1}} \put(18,0.35){\makebox(0,0){n}}
\put(1.5,1.8){\makebox(0,0){1}} \put(4.5,1.8){\makebox(0,0){2}}
\put(13.5,1.8){\makebox(0,0){n-2}}
\put(16.5,1.8){\makebox(0,0){n-1}} \put(1.5,-1.8){\makebox(0,0){1}}
\put(4.5,-1.8){\makebox(0,0){2}} \put(13.5,-1.8){\makebox(0,0){n-2}}
\put(16.5,-1.8){\makebox(0,0){n-1}}
\multiput(0,0)(3,0){6}{\vector(1,1){1.5}}
\multiput(0,0)(3,0){6}{\vector(1,-1){1.5}}
\multiput(1.5,1.5)(3,0){6}{\vector(1,-1){1.5}}
\multiput(1.5,-1.5)(3,0){6}{\vector(1,1){1.5}}
\qbezier(1.5,1.5)(1.5,0)(4.5,-1.5)
\qbezier(4.5,1.5)(4.5,0)(7.5,-1.5)
\qbezier(7.5,1.5)(7.5,0)(10.5,-1.5)
\qbezier(10.5,1.5)(10.5,0)(13.5,-1.5)
\qbezier(13.5,1.5)(13.5,0)(16.5,-1.5)
\multiput(4.5,-1.5)(3,0){5}{\vector(2,-1){0}}
\qbezier(1.5,-1.5)(4.5,0)(4.5,1.5)
\qbezier(4.5,-1.5)(7.5,0)(7.5,1.5)
\qbezier(7.5,-1.5)(10.5,0)(10.5,1.5)
\qbezier(10.5,-1.5)(13.5,0)(13.5,1.5)
\qbezier(13.5,-1.5)(16.5,0)(16.5,1.5)
\multiput(4.5,1.5)(3,0){5}{\vector(0,1){0}}
\multiput(0,0)(3,0){6}{\vector(1,0){3}}
\multiput(1.5,1.5)(3,0){5}{\vector(1,0){3}}
\multiput(1.5,-1.5)(3,0){5}{\vector(1,0){3}}
\put(1.5,0.25){\makebox(0,0){$b_{1}$}}
\put(16.6,0.25){\makebox(0,0){$b_{\text{n-1}}$}}
\put(3,1.75){\makebox(0,0){$c_{1}$}}
\put(15,1.75){\makebox(0,0){$c_{\text{n-2}}$}}
\put(3,-1.75){\makebox(0,0){$a_{1}$}}
\put(15,-1.75){\makebox(0,0){$a_{\text{n-2}}$}}
\put(0.55,0.95){\makebox(0,0){$e_{1}$}}
\put(15.4,0.95){\makebox(0,0){$e_{\text{2n-3}}$}}
\put(2.4,0.95){\makebox(0,0){$e_{2}$}}
\put(17.56,0.95){\makebox(0,0){$e_{\text{2n-2}}$}}
\put(0.6,-0.95){\makebox(0,0){$d_{1}$}}
\put(16.3,-0.75){\makebox(0,0){$d_{\text{2n-3}}$}}
\put(1.6,-0.95){\makebox(0,0){$d_{2}$}}
\put(17.56,-0.95){\makebox(0,0){$d_{\text{2n-2}}$}}
\put(2.55,-1.15){\makebox(0,0){$f_{1}$}}
\put(14.55,-1.2){\makebox(0,0){$f_{\text{n-2}}$}}
\put(3.45,-1.15){\makebox(0,0){$g_{1}$}}
\put(15.45,-1.2){\makebox(0,0){$g_{\text{n-2}}$}}
\multiput(4,-0.2)(0.2,0){6}{\circle*{0.02}}
\multiput(13,-0.2)(0.2,0){6}{\circle*{0.02}}
\multiput(5.5,1.7)(0.2,0){6}{\circle*{0.02}}
\multiput(11.5,1.7)(0.2,0){6}{\circle*{0.02}}
\multiput(5.5,-1.7)(0.2,0){6}{\circle*{0.02}}
\multiput(11.5,-1.7)(0.2,0){6}{\circle*{0.02}}
\end{picture}
\caption{A full square rhomboid of size $n$.}
\label{fsr_fig1}
\end{figure}
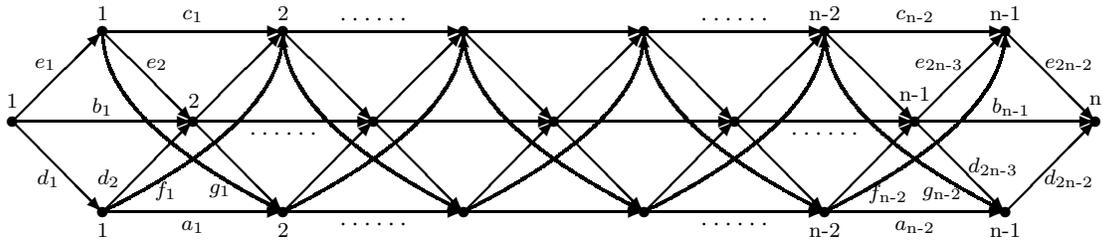

The set of vertices of $N$-vertex $SR$ and $FSR$ consists of $\frac{N+2}{3}$ 
\textit{middle} (\textit{basic}), $\frac{N-1}{3}$ \textit{upper}, and\textit{%
\ }$\frac{N-1}{3}$ \textit{lower} vertices. Upper and lower vertices
numbered $x$ will be denoted in formulae by $\overline{x}$ and $\underline{x}
$, respectively. $SR$ and $FSR$ including $n$ basic vertices will be denoted
by $SR(n)$ and $FSR(n)$, respectively, and will be called an $SR$ and an $%
FSR $ of size $n$.

Some algorithms which generate the expressions of $O\left( n^{\log
_{2}6}\right) $ complexity for $SR(n)$ are discussed in \cite{KoL4}. Our
intention in this paper is to generate and simplify the expressions of full
square rhomboids.

\section{Generating Expressions for Square Rhomboids}

The expressions of square rhomboids are generated using two-vertex
decomposition method (2-VDM) and one-vertex decomposition method (1-VDM) 
\cite{KoL4}. Both methods are based on revealing subgraphs in the initial
graph. The resulting expression is produced by a special composition of
subexpressions describing these subgraphs.

2-VDM is applied as follows. For a non-trivial $SR$ subgraph with a source $%
p $ and a sink $q$ we choose two \textit{decomposition vertices} one of
which belongs to the upper group and the other one belongs to the lower
group. These vertices have the same number $i$ chosen as $\frac{q+p-1}{2}$ ($%
\left\lceil \frac{q+p-1}{2}\right\rceil $ or $\left\lfloor \frac{q+p-1}{2}%
\right\rfloor $). We conditionally split each $SR$ through its decomposition
vertices (see the example in Figure \ref{rhom_fig3}).

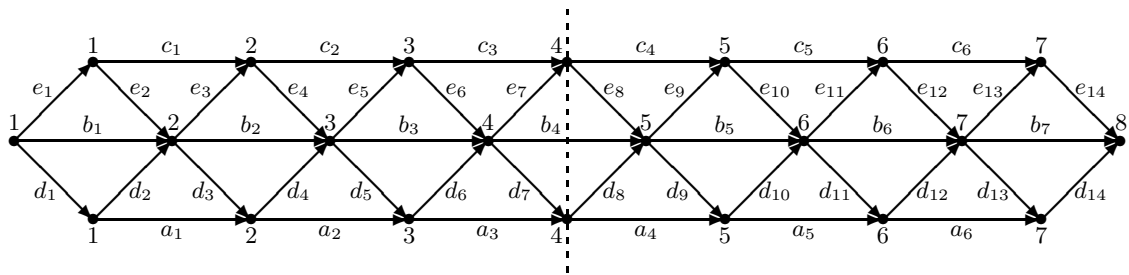
\begin{figure}[tbph]
\setlength{\unitlength}{0.7cm}
\par
\par
\begin{picture}(10,4.0)(2.0,-2.1)\thicklines
\multiput(0,0)(3,0){8}{\circle*{0.1875}}
\multiput(1.5,1.5)(3,0){7}{\circle*{0.1875}}
\multiput(1.5,-1.5)(3,0){7}{\circle*{0.1875}}
\put(0,0.35){\makebox(0,0){1}}
\put(3.02,0.35){\makebox(0,0){2}}
\put(6,0.35){\makebox(0,0){3}}
\put(9,0.35){\makebox(0,0){4}}
\put(12,0.35){\makebox(0,0){5}}
\put(15,0.35){\makebox(0,0){6}}
\put(18,0.35){\makebox(0,0){7}}
\put(21.02,0.35){\makebox(0,0){8}}
\put(1.5,1.8){\makebox(0,0){1}}
\put(4.5,1.8){\makebox(0,0){2}}
\put(7.5,1.8){\makebox(0,0){3}}
\put(10.3,1.8){\makebox(0,0){4}}
\put(13.5,1.8){\makebox(0,0){5}}
\put(16.5,1.8){\makebox(0,0){6}}
\put(19.5,1.8){\makebox(0,0){7}}
\put(1.5,-1.8){\makebox(0,0){1}}
\put(4.5,-1.8){\makebox(0,0){2}}
\put(7.5,-1.8){\makebox(0,0){3}}
\put(10.3,-1.8){\makebox(0,0){4}}
\put(13.5,-1.8){\makebox(0,0){5}}
\put(16.5,-1.8){\makebox(0,0){6}}
\put(19.5,-1.8){\makebox(0,0){7}}
\multiput(0,0)(3,0){7}{\vector(1,1){1.5}}
\multiput(0,0)(3,0){7}{\vector(1,-1){1.5}}
\multiput(1.5,1.5)(3,0){7}{\vector(1,-1){1.5}}
\multiput(1.5,-1.5)(3,0){7}{\vector(1,1){1.5}}
\multiput(0,0)(3,0){7}{\vector(1,0){3}}
\multiput(1.5,1.5)(3,0){6}{\vector(1,0){3}}
\multiput(1.5,-1.5)(3,0){6}{\vector(1,0){3}}
\put(1.5,0.3){\makebox(0,0){$b_{1}$}}
\put(4.5,0.3){\makebox(0,0){$b_{2}$}}
\put(7.5,0.3){\makebox(0,0){$b_{3}$}}
\put(10.2,0.3){\makebox(0,0){$b_{4}$}}
\put(13.5,0.3){\makebox(0,0){$b_{5}$}}
\put(16.5,0.3){\makebox(0,0){$b_{6}$}}
\put(19.5,0.3){\makebox(0,0){$b_{7}$}}
\put(3,1.75){\makebox(0,0){$c_{1}$}}
\put(6,1.75){\makebox(0,0){$c_{2}$}}
\put(9,1.75){\makebox(0,0){$c_{3}$}}
\put(12,1.75){\makebox(0,0){$c_{4}$}}
\put(15,1.75){\makebox(0,0){$c_{5}$}}
\put(18,1.75){\makebox(0,0){$c_{6}$}}
\put(3,-1.75){\makebox(0,0){$a_{1}$}}
\put(6,-1.75){\makebox(0,0){$a_{2}$}}
\put(9,-1.75){\makebox(0,0){$a_{3}$}}
\put(12,-1.75){\makebox(0,0){$a_{4}$}}
\put(15,-1.75){\makebox(0,0){$a_{5}$}}
\put(18,-1.75){\makebox(0,0){$a_{6}$}}
\put(0.55,0.95){\makebox(0,0){$e_{1}$}}
\put(3.55,0.95){\makebox(0,0){$e_{3}$}}
\put(6.55,0.95){\makebox(0,0){$e_{5}$}}
\put(9.55,0.95){\makebox(0,0){$e_{7}$}}
\put(12.55,0.95){\makebox(0,0){$e_{9}$}}
\put(15.5,0.95){\makebox(0,0){$e_{11}$}}
\put(18.5,0.95){\makebox(0,0){$e_{13}$}}
\put(2.4,0.95){\makebox(0,0){$e_{2}$}}
\put(5.4,0.95){\makebox(0,0){$e_{4}$}}
\put(8.4,0.95){\makebox(0,0){$e_{6}$}}
\put(11.4,0.95){\makebox(0,0){$e_{8}$}}
\put(14.45,0.95){\makebox(0,0){$e_{10}$}}
\put(17.45,0.95){\makebox(0,0){$e_{12}$}}
\put(20.45,0.95){\makebox(0,0){$e_{14}$}}
\put(0.6,-0.95){\makebox(0,0){$d_{1}$}}
\put(3.6,-0.95){\makebox(0,0){$d_{3}$}}
\put(6.6,-0.95){\makebox(0,0){$d_{5}$}}
\put(9.6,-0.95){\makebox(0,0){$d_{7}$}}
\put(12.6,-0.95){\makebox(0,0){$d_{9}$}}
\put(15.6,-0.95){\makebox(0,0){$d_{11}$}}
\put(18.6,-0.95){\makebox(0,0){$d_{13}$}}
\put(2.4,-0.95){\makebox(0,0){$d_{2}$}}
\put(5.4,-0.95){\makebox(0,0){$d_{4}$}}
\put(8.4,-0.95){\makebox(0,0){$d_{6}$}}
\put(11.4,-0.95){\makebox(0,0){$d_{8}$}}
\put(14.45,-0.95){\makebox(0,0){$d_{10}$}}
\put(17.45,-0.95){\makebox(0,0){$d_{12}$}}
\put(20.45,-0.95){\makebox(0,0){$d_{14}$}}
\dashline[75]{0.15}(10.5,2.5)(10.5,-2.5)
\end{picture}
\caption{Decomposition of a square rhomboid by 2-VDM.}
\label{rhom_fig3}
\end{figure}

Two kinds of subgraphs are revealed in the graph in the course of
decomposition. The first of them is an $SR$ with a fewer number of vertices
than the initial $SR$. The second one is an $SR$ supplemented by two
additional edges at one of four sides. Possible varieties of this st-dag (we
call it a\textit{\ single-leaf square rhomboid} and denote by $\widehat{SR}$%
) are four subgraphs of an $SR$ in Figure \ref{rhom_fig3} positioned between
vertices $1$ and $\overline{4}$, $\overline{4}$ and $8$, $1$ and $\underline{%
4}$, $\underline{4}$ and $8$. Let $\widehat{SR}(n)$ (an $\widehat{SR}$ of
size $n$) denote an $\widehat{SR}$ including $n$ basic vertices.

We denote by $E(p,q)$ a subexpression related to an $SR$ subgraph with a
source $p$ and a sink $q$. We denote by $E(p,\overline{q})$, $E(\overline {p}%
,q)$, $E(p,\underline{q})$, $E(\underline{p},q)$ subexpressions related to $%
\widehat{SR}$ subgraphs with a source $p$ and a sink $\overline{q}$, a
source $\overline{p}$ and a sink $q$, a source $p$ and a sink $\underline{q}$%
, and a source $\underline{p}$ and a sink $q$.

One can see that any path from vertex $1$ to vertex $8$ in Figure \ref%
{rhom_fig3} passes either through one of decomposition vertices ($\overline{4%
}$ or $\underline{4}$) or through edge $b_{4}$. Therefore, in the general
case a current subgraph is decomposed into six new subgraphs and%
\begin{equation}
E(p,q)\leftarrow E(p,i)b_{i}E(i+1,q)+E(p,\overline{i})E(\overline {i},q)+E(p,%
\underline{i})E(\underline{i},q).  \label{line1}
\end{equation}

Subgraphs described by subexpressions $E(p,i)$ and $E(i+1,q)$ include all
paths from vertex $p$ to vertex $q$ passing through edge $b_{i}$. Subgraphs
described by subexpressions $E(p,\overline{i})$ and $E(\overline{i},q)$
include all paths from vertex $p$ to vertex $q$ passing via vertex $%
\overline{i}$. Subgraphs described by subexpressions $E(p,\underline{i})$
and $E(\underline{i},q)$ include all paths from vertex $p$ to vertex $q$
passing via vertex $\underline{i}$.

An $\widehat{SR}$ subgraph is decomposed into six new subgraphs in the same
way as an $SR$ (see the example in Figure \ref{rhom_fig5}). Two
decomposition vertices (one from the upper and one from the lower group of
vertices) with the same absolute ordinal numbers are selected in the $%
\widehat{SR}$. These vertices are chosen so that the location of the split
is in the middle of the subgraph.

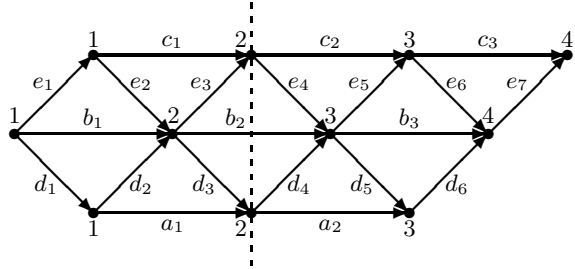
\begin{figure}[tbph]
\setlength{\unitlength}{0.7cm}
\par
\par
\begin{picture}(10,4)(-3.0,-2.0)\thicklines
\multiput(0,0)(3,0){4}{\circle*{0.1875}}
\multiput(1.5,1.5)(3,0){4}{\circle*{0.1875}}
\multiput(1.5,-1.5)(3,0){3}{\circle*{0.1875}}
\put(0,0.35){\makebox(0,0){1}}
\put(3.02,0.35){\makebox(0,0){2}}
\put(6,0.35){\makebox(0,0){3}}
\put(9,0.35){\makebox(0,0){4}}
\put(1.5,1.8){\makebox(0,0){1}}
\put(4.3,1.8){\makebox(0,0){2}}
\put(7.5,1.8){\makebox(0,0){3}}
\put(10.5,1.8){\makebox(0,0){4}}
\put(1.5,-1.8){\makebox(0,0){1}}
\put(4.3,-1.8){\makebox(0,0){2}}
\put(7.5,-1.8){\makebox(0,0){3}}
\multiput(0,0)(3,0){4}{\vector(1,1){1.5}}
\multiput(0,0)(3,0){3}{\vector(1,-1){1.5}}
\multiput(1.5,1.5)(3,0){3}{\vector(1,-1){1.5}}
\multiput(1.5,-1.5)(3,0){3}{\vector(1,1){1.5}}
\multiput(0,0)(3,0){3}{\vector(1,0){3}}
\multiput(1.5,1.5)(3,0){3}{\vector(1,0){3}}
\multiput(1.5,-1.5)(3,0){2}{\vector(1,0){3}}
\put(1.5,0.3){\makebox(0,0){$b_{1}$}}
\put(4.2,0.3){\makebox(0,0){$b_{2}$}}
\put(7.5,0.3){\makebox(0,0){$b_{3}$}}
\put(3,1.75){\makebox(0,0){$c_{1}$}}
\put(6,1.75){\makebox(0,0){$c_{2}$}}
\put(9,1.75){\makebox(0,0){$c_{3}$}}
\put(3,-1.75){\makebox(0,0){$a_{1}$}}
\put(6,-1.75){\makebox(0,0){$a_{2}$}}
\put(0.55,0.95){\makebox(0,0){$e_{1}$}}
\put(3.55,0.95){\makebox(0,0){$e_{3}$}}
\put(6.55,0.95){\makebox(0,0){$e_{5}$}}
\put(9.55,0.95){\makebox(0,0){$e_{7}$}}
\put(2.4,0.95){\makebox(0,0){$e_{2}$}}
\put(5.4,0.95){\makebox(0,0){$e_{4}$}}
\put(8.4,0.95){\makebox(0,0){$e_{6}$}}
\put(0.6,-0.95){\makebox(0,0){$d_{1}$}}
\put(3.6,-0.95){\makebox(0,0){$d_{3}$}}
\put(6.6,-0.95){\makebox(0,0){$d_{5}$}}
\put(2.4,-0.95){\makebox(0,0){$d_{2}$}}
\put(5.4,-0.95){\makebox(0,0){$d_{4}$}}
\put(8.4,-0.95){\makebox(0,0){$d_{6}$}}
\dashline[75]{0.15}(4.5,2.5)(4.5,-2.5)
\end{picture}
\caption{Decomposition of a single-leaf square rhomboid by 2-VDM.}
\label{rhom_fig5}
\end{figure}

Three kinds of subgraphs are revealed in an $\widehat{SR}$ in the course of
decomposition. The first and the second of them are an $SR$ and an $\widehat{%
SR}$, respectively. The third one is an $SR$ supplemented by two additional
pairs of edges (one pair is on the left and another one is on the right).
Possible varieties of this st-dag (we call it a\textit{\ dipterous square
rhomboid} and denote it by $\widehat{\widehat{SR}}$) are illustrated in
Figure \ref{rhom_fig7}(a) (a \textit{parallelogram} $\widehat{\widehat{SR}}$ 
\textit{graph}) and Figure \ref{rhom_fig7}(b) (a \textit{trapezoidal} $%
\widehat{\widehat{SR}}$ \textit{graph}). Let $\widehat{\widehat{SR}}(n)$ (an 
$\widehat{\widehat{SR}}$ of size $n$) denote an $\widehat{\widehat{SR}}$
including $n$ basic vertices.

\begin{figure}[tbph]
\setlength{\unitlength}{0.7cm}
\par
\begin{picture}(10,3.5)(-2.5,-1.5)\thicklines
\put(16.5,0){\makebox(0,0){(b)}}
\multiput(9,0)(3,0){3}{%
\circle*{0.1875}}
\multiput(10.5,1.5)(3,0){3}{\circle*{0.1875}}
\multiput(7.5,-1.5)(3,0){3}{\circle*{0.1875}}
\put(9,0.35){\makebox(0,0){5}}
\put(12,0.35){\makebox(0,0){6}}
\put(15,0.35){\makebox(0,0){7}}
\put(10.3,1.8){\makebox(0,0){5}}
\put(13.5,1.8){\makebox(0,0){6}}
\put(16.5,1.8){\makebox(0,0){7}}
\put(7.5,-1.8){\makebox(0,0){4}}
\put(10.3,-1.8){\makebox(0,0){5}}
\put(13.5,-1.8){\makebox(0,0){6}}
\multiput(9,0)(3,0){3}{\vector(1,1){1.5}}
\multiput(9,0)(3,0){2}{\vector(1,-1){1.5}}
\multiput(10.5,1.5)(3,0){2}{%
\vector(1,-1){1.5}}
\multiput(7.5,-1.5)(3,0){3}{\vector(1,1){1.5}}
\multiput(9,0)(3,0){2}{\vector(1,0){3}}
\multiput(10.5,1.5)(3,0){2}{%
\vector(1,0){3}}
\multiput(7.5,-1.5)(3,0){2}{\vector(1,0){3}}
\put(10.2,0.3){\makebox(0,0){$b_{5}$}}
\put(13.5,0.3){\makebox(0,0){$b_{6}$}}
\put(11.7,1.75){\makebox(0,0){$c_{5}$}}
\put(15,1.75){\makebox(0,0){$c_{6}$}}
\put(9,-1.75){\makebox(0,0){$a_{4}$}}
\put(11.7,-1.75){\makebox(0,0){$a_{5}$}}
\put(9.55,0.95){\makebox(0,0){$e_{9}$}}
\put(12.55,0.95){\makebox(0,0){$e_{11}$}}
\put(15.5,0.95){\makebox(0,0){$e_{13}$}}
\put(11.4,0.95){\makebox(0,0){$e_{10}$}}
\put(14.45,0.95){\makebox(0,0){$e_{12}$}}
\put(9.6,-0.95){\makebox(0,0){$d_{9}$}}
\put(12.6,-0.95){\makebox(0,0){$d_{11}$}}
\put(8.4,-0.95){\makebox(0,0){$d_{8}$}}
\put(11.4,-0.95){\makebox(0,0){$d_{10}$}}
\put(14.45,-0.95){\makebox(0,0){$d_{12}$}}
%
%
\dashline[75]{0.15}(10.5,2.5)(10.5,-2.5)
\end{picture}
\par
\begin{picture}(10,0)(11.7,-2.2)\thicklines
\put(10.5,0){\makebox(0,0){(a)}}
\multiput(12,0)(3,0){4}{\circle*{0.1875}}
\multiput(10.5,1.5)(3,0){5}{\circle*{0.1875}}
\multiput(13.5,-1.5)(3,0){3}{%
\circle*{0.1875}}
\put(12,0.35){\makebox(0,0){6}}
\put(15,0.35){%
\makebox(0,0){7}}
\put(18,0.35){\makebox(0,0){8}}
\put(21.02,0.35){%
\makebox(0,0){9}}
\put(10.5,1.8){\makebox(0,0){5}}
\put(13.5,1.8){%
\makebox(0,0){6}}
\put(16.3,1.8){\makebox(0,0){7}}
\put(19.5,1.8){%
\makebox(0,0){8}}
\put(22.5,1.8){\makebox(0,0){9}}
\put(13.5,-1.8){%
\makebox(0,0){6}}
\put(16.3,-1.8){\makebox(0,0){7}}
\put(19.5,-1.8){%
\makebox(0,0){8}}
\multiput(12,0)(3,0){4}{\vector(1,1){1.5}}
\multiput(12,0)(3,0){3}{\vector(1,-1){1.5}}
\multiput(10.5,1.5)(3,0){4}{%
\vector(1,-1){1.5}}
\multiput(13.5,-1.5)(3,0){3}{\vector(1,1){1.5}}
\multiput(12,0)(3,0){3}{\vector(1,0){3}}
\multiput(10.5,1.5)(3,0){4}{%
\vector(1,0){3}}
\multiput(13.5,-1.5)(3,0){2}{\vector(1,0){3}}
\put(13.5,0.3){\makebox(0,0){$b_{6}$}}
\put(16.5,0.3){%
\makebox(-0.4,0){$b_{7}$}}
\put(19.5,0.3){\makebox(0,0){$b_{8}$}}
\put(12,1.75){\makebox(0,0){$c_{5}$}}
\put(15,1.75){%
\makebox(0,0){$c_{6}$}}
\put(17.7,1.75){\makebox(0,0){$c_{7}$}}
\put(21,1.75){\makebox(0,0){$c_{8}$}}
\put(15,-1.75){%
\makebox(0,0){$a_{6}$}}
\put(17.7,-1.75){\makebox(0,0){$a_{7}$}}
\put(12.55,0.95){\makebox(0,0){$e_{11}$}}
\put(15.5,0.95){%
\makebox(0,0){$e_{13}$}}
\put(18.5,0.95){\makebox(0,0){$e_{15}$}}
\put(21.5,0.95){\makebox(0,0){$e_{17}$}}
\put(11.45,0.95){%
\makebox(0,0){$e_{10}$}}
\put(14.45,0.95){%
\makebox(0,0){$e_{12}$}}
\put(17.45,0.95){\makebox(0,0){$e_{14}$}}
\put(20.45,0.95){\makebox(0,0){$e_{16}$}}
\put(12.6,-0.95){%
\makebox(0,0){$d_{11}$}}
\put(15.6,-0.95){\makebox(0,0){$d_{13}$}}
\put(18.6,-0.95){\makebox(0,0){$d_{15}$}}
\put(14.45,-0.95){%
\makebox(0,0){$d_{12}$}}
\put(17.45,-0.95){\makebox(0,0){$d_{14}$}}
\put(20.45,-0.95){\makebox(0,0){$d_{16}$}}
%
%
\dashline[75]{0.15}(16.5,2.5)(16.5,-2.5)
\end{picture}
\caption{Decomposition of dipterous square rhomboids by 2-VDM.}
\label{rhom_fig7}
\end{figure}
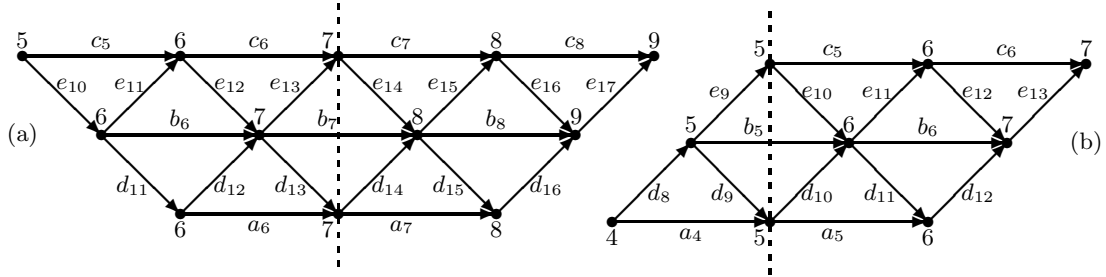

An $\widehat{\widehat{SR}}$ subgraph is decomposed into six new subgraphs in
the same way as an $SR$ and an $\widehat{SR}$ (see examples in Figure \ref%
{rhom_fig7}(a, b)). The number $i$ of the upper and the lower decomposition
vertices for a current $\widehat{\widehat{SR}}$ subgraph positioned between
vertices $p$ and $q$, is chosen as $\frac{q+p}{2}$ ($\left\lceil \frac{q+p}{2%
}\right\rceil $ or $\left\lfloor \frac{q+p}{2}\right\rfloor $). In the
course of decomposition, two kinds of subgraphs are revealed in an $\widehat{%
\widehat{SR}}$. They are an $\widehat{SR}$ and an $\widehat{\widehat{SR}}$.

1-VDM consists in splitting a non-trivial $SR$ with a source $p$ and a sink $%
q$ through one decomposition vertex $i$ located in the basic group of the
subgraph. The number $i$ is chosen as $\frac{q+p}{2}$ ($\left\lceil \frac {%
q+p}{2}\right\rceil $ or $\left\lfloor \frac{q+p}{2}\right\rfloor $) - see
the example in Figure \ref{rhom_fig31}.

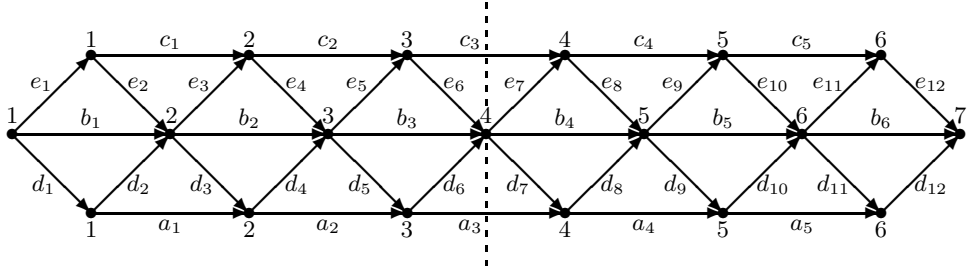
\begin{figure}[h]
\setlength{\unitlength}{0.7cm}
\par
\par
\par
\begin{picture}(10,4.0)(0.5,-2.1)\thicklines
\multiput(0,0)(3,0){7}{\circle*{0.1875}}
\multiput(1.5,1.5)(3,0){6}{\circle*{0.1875}}
\multiput(1.5,-1.5)(3,0){6}{\circle*{0.1875}}
\put(0,0.35){\makebox(0,0){1}} \put(3.02,0.35){\makebox(0,0){2}}
\put(6,0.35){\makebox(0,0){3}} \put(9,0.35){\makebox(0,0){4}}
\put(12,0.35){\makebox(0,0){5}} \put(15,0.35){\makebox(0,0){6}}
\put(18,0.35){\makebox(0,0){7}}
\put(1.5,1.8){\makebox(0,0){1}} \put(4.5,1.8){\makebox(0,0){2}}
\put(7.5,1.8){\makebox(0,0){3}} \put(10.5,1.8){\makebox(0,0){4}}
\put(13.5,1.8){\makebox(0,0){5}} \put(16.5,1.8){\makebox(0,0){6}}
\put(1.5,-1.8){\makebox(0,0){1}} \put(4.5,-1.8){\makebox(0,0){2}}
\put(7.5,-1.8){\makebox(0,0){3}} \put(10.5,-1.8){\makebox(0,0){4}}
\put(13.5,-1.8){\makebox(0,0){5}} \put(16.5,-1.8){\makebox(0,0){6}}
\multiput(0,0)(3,0){6}{\vector(1,1){1.5}}
\multiput(0,0)(3,0){6}{\vector(1,-1){1.5}}
\multiput(1.5,1.5)(3,0){6}{\vector(1,-1){1.5}}
\multiput(1.5,-1.5)(3,0){6}{\vector(1,1){1.5}}
\multiput(0,0)(3,0){6}{\vector(1,0){3}}
\multiput(1.5,1.5)(3,0){5}{\vector(1,0){3}}
\multiput(1.5,-1.5)(3,0){5}{\vector(1,0){3}}
\put(1.5,0.3){\makebox(0,0){$b_{1}$}}
\put(4.5,0.3){\makebox(0,0){$b_{2}$}}
\put(7.5,0.3){\makebox(0,0){$b_{3}$}}
\put(10.5,0.3){\makebox(0,0){$b_{4}$}}
\put(13.5,0.3){\makebox(0,0){$b_{5}$}}
\put(16.5,0.3){\makebox(0,0){$b_{6}$}}
\put(3,1.75){\makebox(0,0){$c_{1}$}}
\put(6,1.75){\makebox(0,0){$c_{2}$}}
\put(8.7,1.75){\makebox(0,0){$c_{3}$}}
\put(12,1.75){\makebox(0,0){$c_{4}$}}
\put(15,1.75){\makebox(0,0){$c_{5}$}}
\put(3,-1.75){\makebox(0,0){$a_{1}$}}
\put(6,-1.75){\makebox(0,0){$a_{2}$}}
\put(8.7,-1.75){\makebox(0,0){$a_{3}$}}
\put(12,-1.75){\makebox(0,0){$a_{4}$}}
\put(15,-1.75){\makebox(0,0){$a_{5}$}}
\put(0.55,0.95){\makebox(0,0){$e_{1}$}}
\put(3.55,0.95){\makebox(0,0){$e_{3}$}}
\put(6.55,0.95){\makebox(0,0){$e_{5}$}}
\put(9.55,0.95){\makebox(0,0){$e_{7}$}}
\put(12.55,0.95){\makebox(0,0){$e_{9}$}}
\put(15.5,0.95){\makebox(0,0){$e_{11}$}}
\put(2.4,0.95){\makebox(0,0){$e_{2}$}}
\put(5.4,0.95){\makebox(0,0){$e_{4}$}}
\put(8.4,0.95){\makebox(0,0){$e_{6}$}}
\put(11.4,0.95){\makebox(0,0){$e_{8}$}}
\put(14.45,0.95){\makebox(0,0){$e_{10}$}}
\put(17.45,0.95){\makebox(0,0){$e_{12}$}}
\put(0.6,-0.95){\makebox(0,0){$d_{1}$}}
\put(3.6,-0.95){\makebox(0,0){$d_{3}$}}
\put(6.6,-0.95){\makebox(0,0){$d_{5}$}}
\put(9.6,-0.95){\makebox(0,0){$d_{7}$}}
\put(12.6,-0.95){\makebox(0,0){$d_{9}$}}
\put(15.6,-0.95){\makebox(0,0){$d_{11}$}}
\put(2.4,-0.95){\makebox(0,0){$d_{2}$}}
\put(5.4,-0.95){\makebox(0,0){$d_{4}$}}
\put(8.4,-0.95){\makebox(0,0){$d_{6}$}}
\put(11.4,-0.95){\makebox(0,0){$d_{8}$}}
\put(14.45,-0.95){\makebox(0,0){$d_{10}$}}
\put(17.45,-0.95){\makebox(0,0){$d_{12}$}}
\dashline[75]{0.15}(9,2.5)(9,0.55)
\dashline[75]{0.15}(9,0.1)(9,-2.5)
\end{picture}
\caption{Decomposition of a square rhomboid by 1-VDM.}
\label{rhom_fig31}
\end{figure}

As for 2-VDM, two $SR$ subgraphs and four $\widehat{SR}$ subgraphs are
revealed in the course of decomposition. Any path from vertex $1$ to vertex $%
7$ in Figure \ref{rhom_fig31} passes through decomposition vertex $4$ or
through edge $c_{3}$ or through edge $a_{3}$. Therefore, in the general case
a current subgraph is decomposed into six new subgraphs and%
\begin{equation}
E(p,q)\leftarrow E(p,i)E(i,q)+E(p,\overline{i-1})c_{i-1}E(\overline{i}%
,q)+E(p,\underline{i-1})a_{i-1}E(\underline{i},q).  \label{line2}
\end{equation}

Subgraphs described by subexpressions $E(p,i)$ and $E(i,q)$ include all
paths from vertex $p$ to vertex $q$ passing through vertex $i$. Subgraphs
described by subexpressions $E(p,\overline{i-1})$ and $E(\overline{i},q)$
include all paths from vertex $p$ to vertex $q$ passing via edge $c_{i-1}$.
Subgraphs described by subexpressions $E(p,\underline{i-1})$ and $E(%
\underline{i},q)$ include all paths from vertex $p$ to vertex $q$ passing
via edge $a_{i-1}$.

An $\widehat{SR}$ subgraph is decomposed by 1-VDM through a decomposition
vertex selected in its basic group into six new subgraphs in a similar way
to 2-VDM. The decomposition vertex is chosen so that the location of the
split is in the middle of the subgraph. The decomposition also gives one $SR$
subgraph, three $\widehat{SR}$ subgraphs, and two $\widehat{\widehat{SR}}$
subgraphs.

Finally, an $\widehat{\widehat{SR}}$ subgraph is also decomposed into two $%
\widehat{SR}$ subgraphs and four $\widehat{\widehat{SR}}$ subgraphs. The
number $i$ of the decomposition vertex in the basic group for a current $%
\widehat{\widehat{SR}}$ subgraph is chosen as $\frac{q+p+1}{2}$ ($%
\left\lceil \frac{q+p+1}{2}\right\rceil $ or $\left\lfloor \frac{q+p+1}{2}%
\right\rfloor $).

Thus by the master theorem \cite{CLR}, the total number of literals $T(n)$
in expressions $Ex(SR(n))$ derived by 2-VDM and 1-VDM is $O\left( n^{\log
_{2}6}\right) $.

However, numerically 1-VDM is more efficient than 2-VDM \cite{KoL4}. It
follows from discussed in \cite{KoL4}\ explicit formulae for $T(n)$ as well (%
$n=2^{k}$ for some positive integer $k\geq 2$). The coefficient of the
leading term of the formulae - $n^{\log _{2}6}$ is equal to $\frac{212}{135}%
\approx 1.57$ for the best algorithm based on 2-VDM and $\frac{154}{135}%
\approx 1.14$ for 1-VDM.

\section{Generating Expressions for Full Square Rhomboids}

Now, we attempt to apply 2-VDM and 1-VDM to a full square rhomboid.

Analogously to graphs mentioned in the previous section, we define \textit{%
single-leaf full square rhomboid} of size $n$ denoted by $\widehat{FSR}(n)$
and \textit{dipterous full square rhomboids} (\textit{trapezoidal} and 
\textit{parallelogram}) of size $n$ denoted by $\widehat{\widehat{FSR}}(n)$.
These graphs, in addition to all edges in corresponding $\widehat{SR}$ and $%
\widehat{\widehat{SR}}$ graphs, have edges labeled by $f$ and $g$ (as in
Figure \ref{fsr_fig1}).

We denote by $E(p,q)$ a subexpression related to an $FSR$ subgraph with a
source $p$ and a sink $q$. We denote by $E(p,\overline{q})$, $E(\overline {p}%
,q)$, $E(p,\underline{q})$, $E(\underline{p},q)$ subexpressions related to $%
\widehat{FSR}$ subgraphs with a source $p$ and a sink $\overline{q}$, a
source $\overline{p}$ and a sink $q$, a source $p$ and a sink $\underline{q}$%
, and a source $\underline{p}$ and a sink $q$. We denote by $E(\overline {p},%
\overline{q})$, $E(\overline{p},\underline{q})$, $E(\underline {p},\overline{%
q})$, $E(\underline{p},\underline{q})$ subexpressions related to $\widehat{%
\widehat{FSR}}$ subgraphs with a source $\overline{p}$ and a sink $\overline{%
q}$, a source $\overline{p}$ and a sink $\underline{q}$, a source $%
\underline{p}$ and a sink $\overline{q}$, and a source $\underline{p}$ and a
sink $\underline{q}$, respectively.

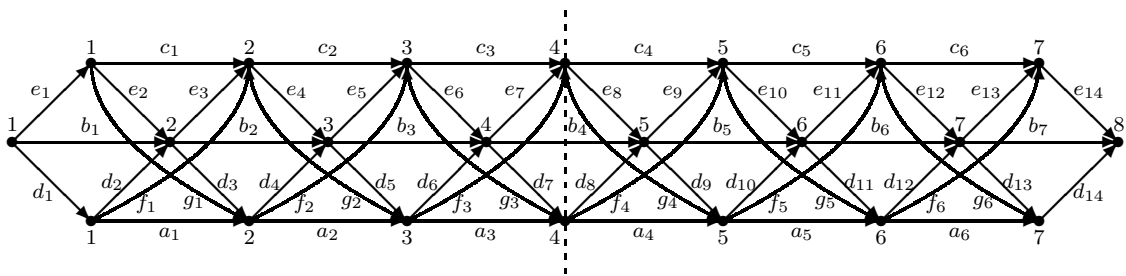
\begin{figure}[tbph]
\setlength{\unitlength}{0.7cm}
\par
\fontsize{8}{0}\selectfont%
\par
\par
\begin{picture}(10,4.0)(2.0,-2.3)\thicklines
\multiput(0,0)(3,0){8}{\circle*{0.1875}}
\multiput(1.5,1.5)(3,0){7}{\circle*{0.1875}}
\multiput(1.5,-1.5)(3,0){7}{\circle*{0.1875}}
\put(0,0.35){\makebox(0,0){1}}
\put(3.02,0.35){\makebox(0,0){2}}
\put(6,0.35){\makebox(0,0){3}}
\put(9,0.35){\makebox(0,0){4}}
\put(12,0.35){\makebox(0,0){5}}
\put(15,0.35){\makebox(0,0){6}}
\put(18,0.35){\makebox(0,0){7}}
\put(21.02,0.35){\makebox(0,0){8}}
\put(1.5,1.8){\makebox(0,0){1}}
\put(4.5,1.8){\makebox(0,0){2}}
\put(7.5,1.8){\makebox(0,0){3}}
\put(10.3,1.8){\makebox(0,0){4}}
\put(13.5,1.8){\makebox(0,0){5}}
\put(16.5,1.8){\makebox(0,0){6}}
\put(19.5,1.8){\makebox(0,0){7}}
\put(1.5,-1.8){\makebox(0,0){1}}
\put(4.5,-1.8){\makebox(0,0){2}}
\put(7.5,-1.8){\makebox(0,0){3}}
\put(10.3,-1.8){\makebox(0,0){4}}
\put(13.5,-1.8){\makebox(0,0){5}}
\put(16.5,-1.8){\makebox(0,0){6}}
\put(19.5,-1.8){\makebox(0,0){7}}
\multiput(0,0)(3,0){7}{\vector(1,1){1.5}}
\multiput(0,0)(3,0){7}{\vector(1,-1){1.5}}
\multiput(1.5,1.5)(3,0){7}{\vector(1,-1){1.5}}
\multiput(1.5,-1.5)(3,0){7}{\vector(1,1){1.5}}
\multiput(0,0)(3,0){7}{\vector(1,0){3}}
\multiput(1.5,1.5)(3,0){6}{\vector(1,0){3}}
\multiput(1.5,-1.5)(3,0){6}{\vector(1,0){3}}
\put(1.5,0.3){\makebox(0,0){$b_{1}$}}
\put(4.5,0.3){\makebox(0,0){$b_{2}$}}
\put(7.5,0.3){\makebox(0,0){$b_{3}$}}
\put(10.75,0.3){\makebox(0,0){$b_{4}$}} %
\put(13.5,0.3){\makebox(0,0){$b_{5}$}}
\put(16.5,0.3){\makebox(0,0){$b_{6}$}}
\put(19.5,0.3){\makebox(0,0){$b_{7}$}}
\put(3,1.75){\makebox(0,0){$c_{1}$}}
\put(6,1.75){\makebox(0,0){$c_{2}$}}
\put(9,1.75){\makebox(0,0){$c_{3}$}}
\put(12,1.75){\makebox(0,0){$c_{4}$}}
\put(15,1.75){\makebox(0,0){$c_{5}$}}
\put(18,1.75){\makebox(0,0){$c_{6}$}}
\put(3,-1.75){\makebox(0,0){$a_{1}$}}
\put(6,-1.75){\makebox(0,0){$a_{2}$}}
\put(9,-1.75){\makebox(0,0){$a_{3}$}}
\put(12,-1.75){\makebox(0,0){$a_{4}$}}
\put(15,-1.75){\makebox(0,0){$a_{5}$}}
\put(18,-1.75){\makebox(0,0){$a_{6}$}}
\put(0.55,0.95){\makebox(0,0){$e_{1}$}}
\put(3.55,0.95){\makebox(0,0){$e_{3}$}}
\put(6.55,0.95){\makebox(0,0){$e_{5}$}}
\put(9.55,0.95){\makebox(0,0){$e_{7}$}}
\put(12.55,0.95){\makebox(0,0){$e_{9}$}}
\put(15.5,0.95){\makebox(0,0){$e_{11}$}}
\put(18.5,0.95){\makebox(0,0){$e_{13}$}}
\put(2.4,0.95){\makebox(0,0){$e_{2}$}}
\put(5.4,0.95){\makebox(0,0){$e_{4}$}}
\put(8.4,0.95){\makebox(0,0){$e_{6}$}}
\put(11.4,0.95){\makebox(0,0){$e_{8}$}}
\put(14.45,0.95){\makebox(0,0){$e_{10}$}}
\put(17.45,0.95){\makebox(0,0){$e_{12}$}}
\put(20.45,0.95){\makebox(0,0){$e_{14}$}}
\put(0.6,-0.95){\makebox(0,0){$d_{1}$}}
\put(4.1,-0.75){\makebox(0,0){$d_{3}$}}
\put(7.1,-0.75){\makebox(0,0){$d_{5}$}}
\put(10.1,-0.75){\makebox(0,0){$d_{7}$}}
\put(13.1,-0.75){\makebox(0,0){$d_{9}$}}
\put(16.1,-0.75){\makebox(0,0){$d_{11}$}}
\put(19.1,-0.75){\makebox(0,0){$d_{13}$}}
\put(1.9,-0.75){\makebox(0,0){$d_{2}$}}
\put(4.9,-0.75){\makebox(0,0){$d_{4}$}}
\put(7.9,-0.75){\makebox(0,0){$d_{6}$}}
\put(10.9,-0.75){\makebox(0,0){$d_{8}$}}
\put(13.85,-0.75){\makebox(0,0){$d_{10}$}}
\put(16.85,-0.75){\makebox(0,0){$d_{12}$}}
\put(20.45,-0.95){\makebox(0,0){$d_{14}$}}
\put(2.55,-1.18){\makebox(0,0){$f_{1}$}}
\put(5.55,-1.18){\makebox(0,0){$f_{2}$}}
\put(8.55,-1.18){\makebox(0,0){$f_{3}$}}
\put(11.55,-1.18){\makebox(0,0){$f_{4}$}}
\put(14.55,-1.18){\makebox(0,0){$f_{5}$}}
\put(17.55,-1.18){\makebox(0,0){$f_{6}$}}
\put(3.45,-1.15){\makebox(0,0){$g_{1}$}}
\put(6.45,-1.15){\makebox(0,0){$g_{2}$}}
\put(9.45,-1.15){\makebox(0,0){$g_{3}$}}
\put(12.45,-1.15){\makebox(0,0){$g_{4}$}}
\put(15.45,-1.15){\makebox(0,0){$g_{5}$}}
\put(18.45,-1.15){\makebox(0,0){$g_{6}$}}
\qbezier(1.5,1.5)(1.5,0)(4.5,-1.5)
\qbezier(4.5,1.5)(4.5,0)(7.5,-1.5)
\qbezier(7.5,1.5)(7.5,0)(10.5,-1.5)
\qbezier(10.5,1.5)(10.5,0)(13.5,-1.5)
\qbezier(13.5,1.5)(13.5,0)(16.5,-1.5)
\qbezier(16.5,1.5)(16.5,0)(19.5,-1.5)
\multiput(4.5,-1.5)(3,0){6}{\vector(2,-1){0}}
\qbezier(1.5,-1.5)(4.5,0)(4.5,1.5)
\qbezier(4.5,-1.5)(7.5,0)(7.5,1.5)
\qbezier(7.5,-1.5)(10.5,0)(10.5,1.5)
\qbezier(10.5,-1.5)(13.5,0)(13.5,1.5)
\qbezier(13.5,-1.5)(16.5,0)(16.5,1.5)
\qbezier(16.5,-1.5)(19.5,0)(19.5,1.5)
\multiput(4.5,1.5)(3,0){6}{\vector(0,1){0}}
\dashline[75]{0.15}(10.5,2.5)(10.5,-2.5)
\end{picture}
\caption{Decomposition of a full square rhomboid by 2-VDM.}
\label{fsr_fig3}
\end{figure}

Figure \ref{fsr_fig3} illustrates the example of decomposition of an $FSR$
by 2-VDM. Appearance of new edges does not change the essence of the
splitting procedure because these edges (labeled by $f$ and $g$) do not
cross the "splitting line" that passes between vertices $\overline{4}$ and $%
\underline{4}$. In all revealed subgraphs, edges $f_{i}$ and $g_{i}$ are
together with pairs of edges $d_{2i}$ and $e_{2i+1}$, and $e_{2i}$ and $%
d_{2i+1}$, respectively. Therefore, in the general case a current $FSR$
subgraph is decomposed into two $FSR$ subgraphs and four $\widehat{FSR}$
subgraphs, and its expression is the same as in statement (\ref{line1}).

$\widehat{FSR}$ and $\widehat{\widehat{FSR}}$ subgraphs are also decomposed
by 2-VDM into six new subgraphs in a similar way to $\widehat{SR}$ and $%
\widehat{\widehat{SR}}$ subgraphs. Thus the complexity of the expression $%
Ex(FSR(n))$ derived by 2-VDM is also $O\left( n^{\log _{2}6}\right) $. 
\begin{figure}[tbph]
\setlength{\unitlength}{0.7cm}
\par
\fontsize{8}{0}\selectfont%
\par
\begin{picture}(10,4.0)(0.5,-2.2)\thicklines
\multiput(0,0)(3,0){7}{\circle*{0.1875}}
\multiput(1.5,1.5)(3,0){6}{\circle*{0.1875}}
\multiput(1.5,-1.5)(3,0){6}{\circle*{0.1875}}
\put(0,0.35){\makebox(0,0){1}}
\put(3.02,0.35){\makebox(0,0){2}}
\put(6,0.35){\makebox(0,0){3}}
\put(9,0.35){\makebox(0,0){4}}
\put(12,0.35){\makebox(0,0){5}}
\put(15,0.35){\makebox(0,0){6}}
\put(18,0.35){\makebox(0,0){7}}
\put(1.5,1.8){\makebox(0,0){1}}
\put(4.5,1.8){\makebox(0,0){2}}
\put(7.5,1.8){\makebox(0,0){3}}
\put(10.3,1.8){\makebox(0,0){4}}
\put(13.5,1.8){\makebox(0,0){5}}
\put(16.5,1.8){\makebox(0,0){6}}
\put(1.5,-1.8){\makebox(0,0){1}}
\put(4.5,-1.8){\makebox(0,0){2}}
\put(7.5,-1.8){\makebox(0,0){3}}
\put(10.3,-1.8){\makebox(0,0){4}}
\put(13.5,-1.8){\makebox(0,0){5}}
\put(16.5,-1.8){\makebox(0,0){6}}
\multiput(0,0)(3,0){6}{\vector(1,1){1.5}}
\multiput(0,0)(3,0){6}{\vector(1,-1){1.5}}
\multiput(1.5,1.5)(3,0){6}{\vector(1,-1){1.5}}
\multiput(1.5,-1.5)(3,0){6}{\vector(1,1){1.5}}
\multiput(0,0)(3,0){6}{\vector(1,0){3}}
\multiput(1.5,1.5)(3,0){5}{\vector(1,0){3}}
\multiput(1.5,-1.5)(3,0){5}{\vector(1,0){3}}
\put(1.5,0.3){\makebox(0,0){$b_{1}$}}
\put(4.5,0.3){\makebox(0,0){$b_{2}$}}
\put(7.5,0.3){\makebox(0,0){$b_{3}$}}
\put(10.5,0.3){\makebox(0,0){$b_{4}$}}
\put(13.5,0.3){\makebox(0,0){$b_{5}$}}
\put(16.5,0.3){\makebox(0,0){$b_{6}$}}
\put(3,1.75){\makebox(0,0){$c_{1}$}}
\put(6,1.75){\makebox(0,0){$c_{2}$}}
\put(8.7,1.75){\makebox(0,0){$c_{3}$}}
\put(12,1.75){\makebox(0,0){$c_{4}$}}
\put(15,1.75){\makebox(0,0){$c_{5}$}}
\put(3,-1.75){\makebox(0,0){$a_{1}$}}
\put(6,-1.75){\makebox(0,0){$a_{2}$}}
\put(8.7,-1.75){\makebox(0,0){$a_{3}$}}
\put(12,-1.75){\makebox(0,0){$a_{4}$}}
\put(15,-1.75){\makebox(0,0){$a_{5}$}}
\put(0.55,0.95){\makebox(0,0){$e_{1}$}}
\put(3.55,0.95){\makebox(0,0){$e_{3}$}}
\put(6.55,0.95){\makebox(0,0){$e_{5}$}}
\put(9.55,0.95){\makebox(0,0){$e_{7}$}}
\put(12.55,0.95){\makebox(0,0){$e_{9}$}}
\put(15.5,0.95){\makebox(0,0){$e_{11}$}}
\put(2.4,0.95){\makebox(0,0){$e_{2}$}}
\put(5.4,0.95){\makebox(0,0){$e_{4}$}}
\put(8.4,0.95){\makebox(0,0){$e_{6}$}}
\put(11.4,0.95){\makebox(0,0){$e_{8}$}}
\put(14.45,0.95){\makebox(0,0){$e_{10}$}}
\put(17.45,0.95){\makebox(0,0){$e_{12}$}}
\put(0.6,-0.95){\makebox(0,0){$d_{1}$}}
\put(4.1,-0.75){\makebox(0,0){$d_{3}$}}
\put(7.1,-0.75){\makebox(0,0){$d_{5}$}}
\put(10.1,-0.75){\makebox(0,0){$d_{7}$}}
\put(13.1,-0.75){\makebox(0,0){$d_{9}$}}
\put(16.1,-0.75){\makebox(0,0){$d_{11}$}}
\put(1.9,-0.75){\makebox(0,0){$d_{2}$}}
\put(4.9,-0.75){\makebox(0,0){$d_{4}$}}
\put(7.9,-0.75){\makebox(0,0){$d_{6}$}}
\put(10.9,-0.75){\makebox(0,0){$d_{8}$}}
\put(13.85,-0.75){\makebox(0,0){$d_{10}$}}
\put(17.45,-0.95){\makebox(0,0){$d_{12}$}}
%
\put(2.55,-1.18){\makebox(0,0){$f_{1}$}}
\put(5.55,-1.18){\makebox(0,0){$f_{2}$}}
\put(8.55,-1.18){\makebox(0,0){$f_{3}$}}
\put(11.55,-1.18){\makebox(0,0){$f_{4}$}}
\put(14.55,-1.18){\makebox(0,0){$f_{5}$}}
\put(3.45,-1.15){\makebox(0,0){$g_{1}$}}
\put(6.45,-1.15){\makebox(0,0){$g_{2}$}}
\put(9.45,-1.15){\makebox(0,0){$g_{3}$}}
\put(12.45,-1.15){\makebox(0,0){$g_{4}$}}
\put(15.45,-1.15){\makebox(0,0){$g_{5}$}}
%
\qbezier(1.5,1.5)(1.5,0)(4.5,-1.5)
\qbezier(4.5,1.5)(4.5,0)(7.5,-1.5)
\qbezier(7.5,1.5)(7.5,0)(10.5,-1.5)
\qbezier(10.5,1.5)(10.5,0)(13.5,-1.5)
\qbezier(13.5,1.5)(13.5,0)(16.5,-1.5)
\multiput(4.5,-1.5)(3,0){5}{\vector(2,-1){0}}
\qbezier(1.5,-1.5)(4.5,0)(4.5,1.5)
\qbezier(4.5,-1.5)(7.5,0)(7.5,1.5)
\qbezier(7.5,-1.5)(10.5,0)(10.5,1.5)
\qbezier(10.5,-1.5)(13.5,0)(13.5,1.5)
\qbezier(13.5,-1.5)(16.5,0)(16.5,1.5)
\multiput(4.5,1.5)(3,0){5}{\vector(0,1){0}}
\dashline[75]{0.15}(9,2.5)(9,0.55)
\dashline[75]{0.15}(9,0.1)(9,-2.5)
\end{picture}
\caption{Decomposition of a full square rhomboid by 1-VDM.}
\label{fsr_fig31}
\end{figure}
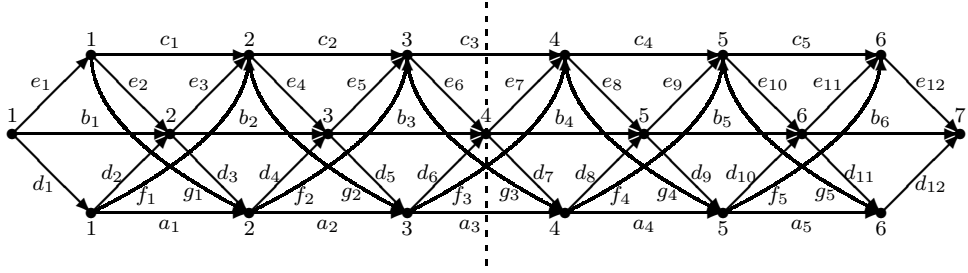

Now, consider decomposition of an $FSR$ by 1-VDM (see the example in Figure %
\ref{fsr_fig31}). One can see that edges $f_{3}$ and $g_{3}$ cross the
"splitting line" that passes through vertex $4$. Hence, any path from vertex 
$1$ to vertex $7$ passes through decomposition vertex $4$ or through one of
the following edges: $c_{3}$, $a_{3}$, $g_{3}$, $f_{3}$. Therefore, in the
general case%
\begin{align}
E(p,q)& \leftarrow E(p,i)E(i,q)+E(p,\overline{i-1})c_{i-1}E(\overline{i}%
,q)+E(p,\underline{i-1})a_{i-1}E(\underline{i},q)+  \label{line3} \\
& E(p,\overline{i-1})g_{i-1}E(\underline{i},q)+E(p,\underline{i-1})f_{i-1}E(%
\overline{i},q).  \notag
\end{align}

Additional parts $E(p,\overline{i-1})g_{i-1}E(\underline{i},q)$ and $E(p,%
\underline{i-1})f_{i-1}E(\overline{i},q)$ which are absent in statement (\ref%
{line2}), describe all paths from vertex $p$ to vertex $q$ passing via edges 
$g_{i-1}$ and $f_{i-1}$, respectively.

Hence, the expression of a current subgraph of size $n$ derived by 1-VDM
includes ten subexpressions related to subgraphs of size $n^{\prime }\approx 
$ $\frac{n}{2}$. This expression can be simplified by putting subexpressions
which appear twice outside the brackets. Finally, statement (\ref{line3})
may be presented as%
\begin{align}
E(p,q)& \leftarrow E(p,i)E(i,q)+E(p,\overline{i-1})\left( c_{i-1}E(\overline{%
i},q)+g_{i-1}E(\underline{i},q)\right) +  \label{line4} \\
& E(p,\underline{i-1})\left( a_{i-1}E(\underline{i},q)+f_{i-1}E(\overline{i}%
,q)\right) ,  \notag
\end{align}%
i.e., the resulting expression for $FSR$ consists of eight subexpressions
related to subgraphs of size $n^{\prime }\approx $ $\frac{n}{2}$ and four
additional literals. The expressions for $\widehat{FSR}$ and $\widehat{%
\widehat{FSR}}$ are constructed in the same way. Thus by the master theorem,
the complexity of the expression $Ex(FSR(n))$ derived by 1-VDM is $O\left(
n^{3}\right) $.

Therefore, 2-VDM that numerically is less efficient than 1-VDM for square
rhomboids is significantly more efficient for full square rhomboids. For
this reason we compute $Ex(FSR)$ by the following recursive relations based
on 2-VDM:

\begin{enumerate}
\item \label{fda_1}$E(p,p)=1$

\item $E(p,\overline{p})=e_{2p-1}$

\item $E(p,\underline{p})=d_{2p-1}$

\item $E(\overline{p},p+1)=e_{2p}$

\item $E(\underline{p},p+1)=d_{2p}$

\item \label{fda_6}$E(\overline{p},\overline{p+1})=c_{p}+e_{2p}e_{2p+1}$

\item $E(\overline{p},\underline{p+1})=e_{2p}d_{2p+1}+g_{p}$

\item $E(\underline{p},\overline{p+1})=d_{2p}e_{2p+1}+f_{p}$

\item \label{fda_9}$E(\underline{p},\underline{p+1})=a_{p}+d_{2p}d_{2p+1}$

\item \label{fda_10}$E(p,\overline{q})=E(p,i)b_{i}E(i+1,\overline{q})+E(p,%
\overline{i})E(\overline{i},\overline{q})+E(p,\underline{i})E(\underline{i},%
\overline{q})$, $i=\left\lfloor \frac{q+p}{2}\right\rfloor \,(q>\nolinebreak
p)\smallskip $

\item \label{fda_11}$E(p,\underline{q})=E(p,i)b_{i}E(i+1,\underline{q})+E(p,%
\overline{i})E(\overline{i},\underline{q})+E(p,\underline{i})E(\underline{i},%
\underline{q})$, $i=\left\lfloor \frac{q+p}{2}\right\rfloor $\thinspace $%
(q>\nolinebreak p)$

\item $E(\overline{p},q)=E(\overline{p},i)b_{i}E(i+1,q)+E(p,\overline{i})E(%
\overline{i},q)+E(\overline{p},\underline{i})E(\underline{i},q)$, $%
i=\left\lfloor \frac{q+p}{2}\right\rfloor $\thinspace $(q>\nolinebreak
p+\nolinebreak 1)\smallskip $

\item $E(\underline{p},q)=E(\underline{p},i)b_{i}E(i+1,q)+E(\underline{p},%
\overline{i})E(\overline{i},q)+E(\underline{p},\underline{i})E(\underline{i}%
,q)$, $i=\left\lfloor \frac{q+p}{2}\right\rfloor $\thinspace $%
(q>\nolinebreak p+\nolinebreak 1)\smallskip $

\item \label{fda_14}$E(\overline{p},\overline{q})=E(\overline{p}%
,i)b_{i}E(i+1,\overline{q})+E(\overline{p},\overline{i})E(\overline{i},%
\overline{q})+E(\overline{p},\underline{i})E(\underline{i},\overline{q})$, $%
i=\frac{q+p}{2}\,(q>\nolinebreak p+\nolinebreak 1)\smallskip $

\item $E(\overline{p},\underline{q})=E(\overline{p},i)b_{i}E(i+1,\underline{q%
})+E(\overline{p},\overline{i})E(\overline{i},\underline{q})+E(\overline{p},%
\underline{i})E(\underline{i},\underline{q})$, $i=\frac{q+p}{2}%
\,(q>\nolinebreak p+\nolinebreak 1)\smallskip $

\item $E(\underline{p},\overline{q})=E(\underline{p},i)b_{i}E(i+1,\overline{q%
})+E(\underline{p},\overline{i})E(\overline{i},\overline{q})+E(\underline{p},%
\underline{i})E(\underline{i},\overline{q})$, $i=\frac{q+p}{2}%
\,(q>\nolinebreak p+\nolinebreak 1)\smallskip $

\item \label{fda_17}$E(\underline{p},\underline{q})=E(\underline{p}%
,i)b_{i}E(i+1,\underline{q})+E(\underline{p},\overline{i})E(\overline{i},%
\underline{q})+E(\underline{p},\underline{i})E(\underline{i},\underline{q})$%
, $i=\frac{q+p}{2}\,(q>\nolinebreak p+\nolinebreak 1)\smallskip $

\item \label{fda_18}$E(p,q)=E(p,i)b_{i}E(i+1,q)+E(p,\overline{i})E(\overline{%
i},q)+E(p,\underline{i})E(\underline{i},q)$, $i=\frac{q+p-1}{2}%
\,(q>\nolinebreak p)$.
\end{enumerate}

The following lemma results from relations \ref{fda_6} -- \ref{fda_9} and %
\ref{fda_14} -- \ref{fda_17}.

\begin{lemma}
\label{lem_fsr_2-vdm}Complexities of expressions $Ex\left( trapezoidal\text{ }%
\widehat{\widehat{FSR}}(n)\right) $ and\linebreak $Ex\left( parallelogram%
\text{ }\widehat{\widehat{FSR}}(n)\right) $ derived by 2-VDM are equal.
\end{lemma}

\begin{proof}
According to relations \ref{fda_6} -- \ref{fda_9}, $Ex\left( \widehat{%
\widehat{FSR}}\left( 1\right) \right) $ contains three literals for
trapezoidal and parallelogram $\widehat{\widehat{FSR}}$ graphs. Expressions $%
Ex\left( trapezoidal\text{ }\widehat{\widehat{FSR}}\left( n\right) \right) $
and $Ex\left( parallelogram\text{ }\widehat{\widehat{FSR}}\left( n\right)
\right) $ consist of the same components (see relations \ref{fda_14} -- \ref%
{fda_17}) for $n>1$. They are literal $b_{i}$ and the following six
subexpressions:

$Ex\left( \widehat{FSR}\left( \left\lceil \frac{n}{2}\right\rceil \right)
\right) $; $Ex\left( \widehat{FSR}\left( \left\lfloor \frac{n}{2}%
\right\rfloor \right) \right) $;

$Ex\left( trapezoidal\text{ }\widehat{\widehat{FSR}}\left( \left\lceil \frac{%
n}{2}\right\rceil \right) \right) $; $Ex\left( trapezoidal\text{ }\widehat{%
\widehat{FSR}}\left( \left\lfloor \frac{n}{2}\right\rfloor \right) \right) $;

$Ex\left( parallelogram\text{ }\widehat{\widehat{FSR}}\left( \left\lceil 
\frac{n}{2}\right\rceil \right) \right) $; $Ex\left( parallelogram\text{ }%
\widehat{\widehat{FSR}}\left( \left\lfloor \frac{n}{2}\right\rfloor \right)
\right) $.

The subexpression of each kind appears once in $Ex\left( trapezoidal\text{ }%
\widehat{\widehat{FSR}}\right) $ and once in $Ex\left( parallelogram\text{ }%
\widehat{\widehat{FSR}}\right) $. Hence, the expression complexity for any $%
\widehat{\widehat{FSR}}\left( n\right) $ is equal to the sum of complexities
of subexpressions above increased by one. For this reason, complexities of
expressions $Ex\left( trapezoidal\text{ }\widehat{\widehat{FSR}}\left(
n\right) \right) $ and $Ex\left( parallelogram\text{ }\widehat{\widehat{FSR}}%
\left( n\right) \right) $are equal.
\end{proof}

The following proposition results from Lemma \ref{lem_fsr_2-vdm} and
relations \ref{fda_1} -- \ref{fda_18}.

\begin{proposition}
\label{th_fda}The total number of literals $T(n)$ in the expression
$Ex(FSR(n))$ derived by 2-VDM is defined recursively as follows:

1) $T(1)=0;\ $2) $\widehat{T}(1)=1;\ $3) $\widehat{\widehat{T}}(1)=3$

4) $T(n)=T\left(  \left\lceil \frac{n}{2}\right\rceil \right)  +T\left(
\left\lfloor \frac{n}{2}\right\rfloor \right)  +2\widehat{T}\left(
\left\lceil \frac{n}{2}\right\rceil \right)  +2\widehat{T}\left(  \left\lfloor
\frac{n}{2}\right\rfloor \right)  +1\quad(n>1)$

5) $\widehat{T}(n)=T\left(  \left\lceil \frac{n}{2}\right\rceil \right)
+\widehat{T}\left(  \left\lfloor \frac{n}{2}\right\rfloor \right)
+2\widehat{T}\left(  \left\lceil \frac{n}{2}\right\rceil \right)
+2\widehat{\widehat{T}}\left(  \left\lfloor \frac{n}{2}\right\rfloor \right)
+1\quad(n>1)$

6) $\widehat{\widehat{T}}(n)=\widehat{T}\left(  \left\lceil \frac{n}%
{2}\right\rceil \right)  +\widehat{T}\left(  \left\lfloor \frac{n}%
{2}\right\rfloor \right)  +2\widehat{\widehat{T}}\left(  \left\lceil \frac
{n}{2}\right\rceil \right)  +2\widehat{\widehat{T}}\left(  \left\lfloor
\frac{n}{2}\right\rfloor \right)  +1\quad(n>1)$,\smallskip

where $\widehat{T}(n)$ and $\widehat{\widehat{T}}(n)$ are the total numbers of
literals in $Ex(\widehat{FSR}(n))$ and $Ex\left(  \widehat{\widehat{FSR}%
}\left(  n\right)  \right)  $, respectively.
\end{proposition}

\begin{proof}
Initial formulae (1 -- 3) follow directly from relations \ref{fda_1} -- \ref%
{fda_9} of 2-VDM. General formulae (4 -- 6) are based on the structure of
expressions \ref{fda_10} -- \ref{fda_18} of 2-VDM and on Lemma \ref%
{lem_fsr_2-vdm}. Indeed, the location of the split is in the middle of all
subgraphs, expressions \ref{fda_10} -- \ref{fda_18} include one additional
literal ($b_{i}$) and, by Lemma \ref{lem_fsr_2-vdm}, complexities of
expressions $Ex\left( trapezoidal\text{ }\widehat{\widehat{FSR}}\left(
n\right) \right) $ and\linebreak $Ex\left( parallelogram\text{ }\widehat{%
\widehat{FSR}}\left( n\right) \right) $ may be denoted equally.
\end{proof}

It is of interest to obtain exact formulae describing complexity of the
expression $Ex(FSR(n))$ derived by 2-VDM. We attempt to do it for $n$ that
is a power of two, i.e., $n=2^{k}$ for some positive integer $k\geq 1$.
Formulae (4 -- 6) of Proposition \ref{th_fda} are presented in this case as 
\begin{equation}
\left\{ 
\begin{array}{l}
T\left( n\right) =2T\left( \frac{n}{2}\right) +4\widehat{T}\left( \frac{n}{2}%
\right) +1 \\ 
\widehat{T}\left( n\right) =T\left( \frac{n}{2}\right) +3\widehat{T}\left( 
\frac{n}{2}\right) +2\widehat{\widehat{T}}\left( \frac{n}{2}\right) +1 \\ 
\widehat{\widehat{T}}\left( n\right) =2\widehat{T}\left( \frac{n}{2}\right)
+4\widehat{\widehat{T}}\left( \frac{n}{2}\right) +1,%
\end{array}%
\right.  \label{fsrf3}
\end{equation}%
respectively. The following explicit formulae for simultaneous recurrences (%
\ref{fsrf3})\textbf{\ }are obtained by the method for linear recurrence
relations solving \cite{Ros}: 
\begin{align*}
T(n)& =\frac{89}{45}n^{\log _{2}6}-\frac{20}{9}n^{\log _{2}3}-\frac{1}{5} \\
\widehat{T}(n)& =\frac{89}{45}n^{\log _{2}6}-\frac{5}{9}n^{\log _{2}3}-\frac{%
1}{5} \\
\widehat{\widehat{T}}(n)& =\frac{89}{45}n^{\log _{2}6}+\frac{10}{9}n^{\log
_{2}3}-\frac{1}{5}.
\end{align*}

\section{A Combined Method for Generating Expressions of Full Square
Rhomboids}

\setlength{\textfloatsep}{0.1cm}

As shown in the previous section, the complexity $T(n)$ of the expression
\linebreak $Ex(FSR(n))$ is $O\left( n^{\log _{2}6}\right) $ if it is derived
by 2-VDM and $O\left( n^{3}\right) $ if 1-VDM is used. However, despite on
the asymptotic advantage of 2-VDM, expressions constructed by 1-VDM are
shorter for some small values of $n$. One can see (Table \ref{tab1}) that
complexities for 1-VDM are smaller than corresponding complexities for 2-VDM
when $n=3$ and, as a result, when $n=5$ and $n=6$. Expressions of graphs
with these sizes are included by expressions of graphs with larger sizes.

\begin{table}[tbp] \centering%
\begin{tabular}{||c||c|c|c||c|c|c||}
\hline\hline
$n$ & $%
\begin{array}{c}
T(n)\text{,} \\ 
\text{1-VDM}%
\end{array}
$ & $%
\begin{array}{c}
\widehat{T}\left( n\right) \text{,} \\ 
\text{1-VDM}%
\end{array}
$ & $%
\begin{array}{c}
\widehat{\widehat{T}}\left( n\right) \text{,} \\ 
\text{1-VDM}%
\end{array}
$ & $%
\begin{array}{c}
T(n)\text{,} \\ 
\text{2-VDM}%
\end{array}
$ & $%
\begin{array}{c}
\widehat{T}\left( n\right) \text{,} \\ 
\text{2-VDM}%
\end{array}
$ & $%
\begin{array}{c}
\widehat{\widehat{T}}\left( n\right) \text{,} \\ 
\text{2-VDM}%
\end{array}
$ \\ \hline\hline
$1$ & $0$ & $1$ & $3$ & $0$ & $1$ & $3$ \\ \hline
$2$ & $5$ & $10$ & $15$ & $5$ & $10$ & $15$ \\ \hline
$3$ & $20$ & $29$ & $42$ & $28$ & $33$ & $48$ \\ \hline
$4$ & $53$ & $66$ & $85$ & $51$ & $66$ & $81$ \\ \hline
$5$ & $104$ & $123$ & $152$ & $120$ & $135$ & $170$ \\ \hline
$6$ & $175$ & $204$ & $243$ & $189$ & $224$ & $259$ \\ \hline
$7$ & $284$ & $323$ & $388$ & $278$ & $313$ & $358$ \\ \hline
$8$ & $409$ & $474$ & $531$ & $367$ & $412$ & $457$ \\ \hline
$9$ & $608$ & $665$ & $760$ & $574$ & $619$ & $704$ \\ \hline
$10$ & $793$ & $888$ & $975$ & $781$ & $866$ & $951$ \\ \hline
$20$ & $6325$ & $6800$ & $7061$ & $5027$ & $5282$ & $5537$ \\ \hline
$30$ & $21351$ & $22326$ & $23301$ & $13077$ & $13472$ & $13867$ \\ \hline
$40$ & $50905$ & $53280$ & $54063$ & $31183$ & $31948$ & $32713$ \\ \hline
$50$ & $98935$ & $100690$ & $102445$ & $54493$ & $55518$ & $56543$ \\ \hline
$60$ & $171195$ & $176070$ & $178995$ & $80043$ & $81228$ & $82413$ \\ 
\hline\hline
\end{tabular}
\caption{Complexities for 1-VDM and 2-VDM.\label{tab1}}%
\end{table}%

\begin{table}[tbp] \centering%
\begin{tabular}{||c||c|c|c||}
\hline\hline
$n$ & $T(n)$ & $\widehat{T}\left( n\right) $ & $\widehat{\widehat{T}}\left(
n\right) $ \\ \hline\hline
$1$ & $0$ & $1$ & $3$ \\ \hline
$2$ & $5$ & $10$ & $15$ \\ \hline
$3$ & $20$ & $29$ & $42$ \\ \hline
$4$ & $51$ & $66$ & $81$ \\ \hline
$5$ & $104$ & $119$ & $152$ \\ \hline
$6$ & $157$ & $192$ & $227$ \\ \hline
$7$ & $262$ & $297$ & $342$ \\ \hline
$8$ & $367$ & $412$ & $457$ \\ \hline
$9$ & $526$ & $571$ & $652$ \\ \hline
$10$ & $685$ & $766$ & $847$ \\ \hline
$20$ & $4435$ & $4678$ & $4921$ \\ \hline
$30$ & $12789$ & $13184$ & $13579$ \\ \hline
$40$ & $27583$ & $28312$ & $29041$ \\ \hline
$50$ & $48445$ & $49470$ & $50495$ \\ \hline
$60$ & $78315$ & $79500$ & $80685$ \\ \hline\hline
\end{tabular}
\caption{Complexities for modified 2-VDM.\label{tab2}}%
\end{table}%

We modify 2-VDM through generating expressions of graphs with size $3$ by
1-VDM and obtain the following new values: $T(5)=104$, $\widehat{T}\left(
5\right) =119$, $\widehat{\widehat{T}}\left( 5\right) =154$, $T(6)=157$, $%
\widehat{T}\left( 6\right) =192$, $\widehat{\widehat{T}}\left( 6\right) =227$%
. So, all new values except $\widehat{\widehat{T}}\left( 5\right) $ are not
greater than corresponding values for 1-VDM presented in Table \ref{tab1}.
For this reason, we additionally improve 2-VDM and derive by 1-VDM\ the
expression $Ex\left( \widehat{\widehat{FSR}}\left( 5\right) \right) $ as
well. The final complexities for modified 2-VDM are presented in Table \ref%
{tab2}. For all $n$ except $n=3$ and $n=5$ they are determined in accordance
with Proposition \ref{th_fda}. In addition, we use the following formulae
which result from statement (\ref{line4}) and from analogous relations for
computing expressions $Ex\left( \widehat{FSR}\right) $ and $Ex\left( 
\widehat{\widehat{FSR}}\right) $ by 1-VDM:

$T(n)=T\left( \left\lfloor \frac{n}{2}\right\rfloor +1\right) +T\left(
\left\lceil \frac{n}{2}\right\rceil \right) +2\widehat{T}\left( \left\lfloor 
\frac{n}{2}\right\rfloor \right) +4\widehat{T}\left( \left\lceil \frac{n}{2}%
\right\rceil -1\right) +4\quad(n=3)$

$\widehat{T}(n)=T\left( \left\lfloor \frac{n}{2}\right\rfloor +1\right) +%
\widehat{T}\left( \left\lceil \frac{n}{2}\right\rceil \right) +4\widehat {T}%
\left( \left\lfloor \frac{n}{2}\right\rfloor \right) +2\widehat {\widehat{T}}%
\left( \left\lceil \frac{n}{2}\right\rceil -1\right) +4\quad(n=3)$

$\widehat{\widehat{T}}(n)=\widehat{T}\left( \left\lfloor \frac{n}{2}%
\right\rfloor +1\right) +\widehat{T}\left( \left\lceil \frac{n}{2}%
\right\rceil \right) +2\widehat{\widehat{T}}\left( \left\lfloor \frac{n}{2}%
\right\rfloor \right) +4\widehat{\widehat{T}}\left( \left\lceil \frac{n}{2}%
\right\rceil -1\right) +4\quad (n=3,5).$

\section{Conclusions and Future Work}

Various non-series-parallel graphs (Fibonacci graphs, square rhomboids, full
square rhomboids, etc.) have expressions with polynomial complexity despite
their relatively complex structure. The existence of a decomposition method
for a graph $G$ is a sufficient condition for the existence of such
expression for $G$. The complexity depends, in particular, on the number of
revealed subgraphs in each recursive step of the decomposition procedure.
Different decomposition methods may be applied to the same class of graphs
and one of the methods may be more efficient for one class and less
efficient for another one.

An undirected graph in which every subgraph has a vertex of degree at most $%
k $ is called \textit{k-inductive} \cite{Ira}. For instance, trees are $1$%
-inductive graphs, and planar graphs are $5$-inductive. \textit{Random
scale-free networks} \cite{BaA} demonstrate important practical examples of $%
k$-inductive graphs. As follows from \cite{ChE}, a graph $G$ is $k$%
-inductive if and only if the edges of $G$ can be oriented to form a
directed acyclic graph with out-degree of its vertices at most $k$. Thus
underlying graphs of Fibonacci graphs are $2$-inductive while underlying
graphs of square and full square rhomboids are $3$-inductive.

We intend to extend the presented decomposition technique to a class of
st-dags whose underlying graphs are $k$-inductive.


\begin{thebibliography}{99}
\bibitem{BaA} A.-L. Barab\'{a}si and R. Albert, \emph{Emergence of Scaling
in Random Networks}, Science, Vol. \textbf{286}, No 5439,1999, 509--512.

\bibitem{BKS} W. W. Bein, J. Kamburowski, and M. F. M. Stallmann, \emph{%
Optimal Reduction of Two-Terminal Directed Acyclic Graphs}, SIAM Journal of
Computing, Vol. \textbf{21}, No 6, 1992, 1112--1129.

\bibitem{ChE} M. Chrobak and D. Eppstein, \emph{Planar Orientations with Low
Out-Degree and Compaction of Adjacency Matrices}, Theoretical Computer
Science Vol. \textbf{86}, No 2, 1991, 243--266.

\bibitem{CLR} T. H. Cormen, C. E. Leiseron, and R. L. Rivest, \emph{%
Introduction to Algorithms}, The MIT Press, Cambridge, Massachusetts, 2001.

\bibitem{Duf} R. J. Duffin, \emph{Topology of Series-Parallel Networks},
Journal of Mathematical Analysis and Applications \textbf{10}, 1965,
303--318.

\bibitem{GoM} M. Ch. Golumbic and A. Mintz,\textit{\ }\emph{Factoring Logic
Functions Using Graph Partitioning}, in: Proc. IEEE/ACM Int. Conf. Computer
Aided Design, 1999, 109--114.

\bibitem{GMR} M. Ch. Golumbic, A. Mintz, and U. Rotics,\textit{\ }\emph{%
Factoring and Recognition of Read-Once Functions using Cographs and Normality%
}, in: Proc. 38th Design Automation Conf., 2001, 195--198.

\bibitem{GoP} M. Ch. Golumbic and Y. Perl,\textit{\ }\emph{Generalized
Fibonacci Maximum Path Graphs}, Discrete Mathematics \textbf{28}, 1979,
237--245.\emph{\ }

\bibitem{Ira} S. Irani, \emph{Coloring Inductive Graphs On-Line},
Algorithmica, Vol. \textbf{11}, No 1, 1994, 53--72.

\bibitem{KoL} M. Korenblit and V. E. Levit, \emph{On Algebraic Expressions
of Series-Parallel and Fibonacci Graphs}, in: Discrete Mathematics and
Theoretical Computer Science, Proc. 4th Int. Conf., DMTCS 2003, LNCS \textbf{%
2731}, Springer, 2003, 215--224.

\bibitem{KoL4} M. Korenblit and V. E. Levit, \emph{A One-Vertex
Decomposition Algorithm for Generating Algebraic Expressions of Square
Rhomboids}, arXiv.org, Cornell University Library, 2012,
http://arxiv.org/abs/1211.1661

\bibitem{Mun1} D. Mundici, \emph{Functions Computed by Monotone Boolean
Formulas with no Repeated Variables}, Theoretical Computer Science \textbf{66%
}, 1989, 113-114.

\bibitem{Mun2} D. Mundici, \emph{Solution of Rota's Problem on the Order of
Series-Parallel Networks}, Advances in Applied Mathematics \textbf{12},
1991, 455--463.

\bibitem{Nau} V. Naumann, \emph{Measuring the Distance to Series-Parallelity
by Path Expressions}, in: Graph-Theoretic Concepts in Computer Science,
Proc. 20th Int. Workshop, WG '94, LNCS \textbf{903}, Springer, 1994,
269--281.

\bibitem{Ros} \emph{Handbook of Discrete and Combinatorial Mathematics},
edited by K. H. Rosen, CRC Press, Boca Raton, 2000.

\bibitem{SaW} P. Savicky and A. R. Woods, \emph{The Number of Boolean
Functions Computed by Formulas of a Given Size}, Random Structures and
Algorithms \textbf{13}, 1998, 349--382.

\bibitem{Wan} A. R. R. Wang, \emph{Algorithms for Multilevel Logic
Optimization}, Ph.D. Thesis, University of California, Berkeley, 1989.
\end{thebibliography}
\end{document}